\documentclass[aps,pra,onecolumn,nofootinbib,superscriptaddress,showpacs]{revtex4}
\usepackage{amsmath,amssymb,amsthm}
\usepackage{amsfonts}
\usepackage{amssymb}
\usepackage{graphicx}
\usepackage{comment}
\usepackage{color}
\usepackage{soul}
\usepackage{natbib}
\newtheorem{theorem}{Theorem}

\newtheorem{prop}[theorem]{Proposition}

\usepackage[title]{appendix}
\usepackage{mathrsfs}

\begin{document}

\title{Entropic Uncertainty Relations for Mutually Unbiased Operator Frames}

\author{Jesni Shamsul Shaari}
\affiliation{Department of Physics, Faculty of Science, International Islamic University Malaysia (IIUM),
Jalan Sultan Ahmad Shah, Bandar Indera Mahkota, 25200 Kuantan, Pahang, Malaysia}
\affiliation{IIUM Photonics and Quantum Centre, International Islamic University Malaysia (IIUM),
Jalan Sultan Ahmad Shah, Bandar Indera Mahkota, 25200 Kuantan, Pahang, Malaysia}

\author{Stefano Mancini}
\affiliation{School of Science \& Technology, University of Camerino, I-62032 Camerino, Italy}
\affiliation{ INFN Sezione di Perugia, I-06123 Perugia, Italy}

\date{\today}

\begin{abstract}
We develop an operator-frame formulation of entropic uncertainty relations in the Hilbert--Schmidt space of operators. For general continuous indexed operator frames, we derive an entropic uncertainty relation for the associated coefficient distributions by combining endpoint norm estimates with Riesz--Thorin interpolation. We then identify a distinguished class of mutually unbiased operator frames, defined through constant-modulus trace overlaps. Under suitable structural conditions, the corresponding coefficient amplitudes are related by a bilinear Fourier transform, leading to a stronger Hirschman--Beckner-type entropic uncertainty relation. As canonical realizations, we consider Weyl displacement operators and Wigner kernels, as well as Cartesian dyadic frames generated by position and momentum eigenstates. These examples recover familiar continuous-variable Fourier dualities while extending entropic uncertainty relations beyond measurement outcomes to operator representations themselves.
\end{abstract}


\maketitle

\section{Introduction}

\noindent The uncertainty principle occupies a central role in quantum theory,
originating from Heisenberg's observation that complementary physical
quantities such as position and momentum cannot be simultaneously
specified with arbitrary precision~\cite{Heisenberg1927}. In its standard formulation, this complementarity
arises from the Fourier relationship between the position-space and
momentum-space wavefunction amplitudes. More precisely, a quantum state whose
wavefunction is sharply localized in position space necessarily
possesses a broadly distributed momentum-space wavefunction, and vice
versa. The resulting uncertainty relations may be
expressed either in terms of variances, as in the
Kennard--Robertson formulation~\cite{Robertson1929}, or,
more generally, through entropic measures associated with the
corresponding probability distributions~\cite{Hirschman1957,Beckner1975,Bialynicki1975}. An extensive review of entropic uncertainty relations and their
applications may be found in Ref.~\cite{Coles2017}. In continuous-variable systems, entropic uncertainty relations have been extensively developed beyond the canonical position--momentum setting, including formulations for arbitrary quadratures, multimode systems, and (unitary) linear canonical transformations~\cite{Cerf1,Cerf2}. Uncertainty relations have also been formulated in the framework of countable frame representations \cite{Ricaud2013}.
In finite-dimensional systems, entropic uncertainty
relations were later formulated in terms of overlaps
between orthonormal bases by Maassen and
Uffink~\cite{Maassen1988}. Mutually unbiased bases (MUBs) \cite{Schwinger,Durt}
correspond to the maximally complementary case,
yielding the strongest entropic lower bounds. Simply put, two orthonormal bases are said to be
mutually unbiased when complete knowledge of a state in one basis
implies maximal uncertainty in the other. In finite-dimensional quantum
systems, mutually unbiased bases have found important applications in
quantum information, quantum tomography, and quantum cryptography. The notion of mutual unbiasedness may also
be extended to continuous-variable systems, where the
position and momentum eigenbases provide a canonical
example~\cite{Weigert2008}. 

The idea of complementarity has also been extended from state vectors to
operator spaces through the notion of mutually unbiased unitary bases
(MUUBs), where unitary operators form orthogonal bases with constant
Hilbert--Schmidt overlaps \cite{Scott2008, JRM2015}. In this setting,
maximal entropic uncertainty between operator
measurements is likewise associated with mutually
unbiased operator structures, closely paralleling the
role played by MUBs for ordinary observables \cite{JM2020, JRM2021}.

Motivated by these developments, in the present work we formulate a general operator-frame approach for continuous-variable systems. Beginning with families of operators acting on a Hilbert space, we derive a general entropic uncertainty relation by establishing suitable endpoint norm estimates and applying the Riesz--Thorin interpolation theorem (see theorem IX.17 of \cite{ReedSimon2}). We then introduce a notion of mutual unbiasedness for operator families, defined through constant-modulus trace overlaps between operator pairings. Under additional conditions ensuring a sufficiently regular overlap structure, this mutual unbiasedness condition induces a generalized Fourier transform relation between the corresponding operator coefficient functions, yielding an explicit realization of the general uncertainty framework. Unlike conventional continuous-variable entropic uncertainty relations, which concern probability distributions associated with measurements of observables, the present formulation is based on coefficient distributions arising from operator expansions. The resulting uncertainty relation expresses the impossibility of simultaneously constructing operator representations that are sharply localized in mutually unbiased operator frames. While the initial derivation shares certain interpolation-theoretic ingredients with frame-based uncertainty relations \cite{Ricaud2013}, the present framework is formulated in terms of continuous operator frames and culminates in an operator-level notion of mutual unbiasedness. More precisely, the continuous parameterization of the operator families allows us to consider continuous translational covariance in Sec.~III; together with the constant-modulus overlap condition, it leads to the bilinear phase structure and the resulting Fourier-dual relations.
This then enables the derivation of Fourier-type entropic uncertainty relations, yielding stronger lower bounds than those obtained in the general operator-frame setting.

This continuous structure also allows us to 
specialize the general framework to two physically relevant
continuous-variable settings. First, choosing an antisymmetric
bilinear form on phase space yields the symplectic Fourier structure
associated with displacement operators and Wigner kernels \cite{CahillGlauber19691,CahillGlauber19692,Royer1977}. Second, choosing a symmetric bilinear form in Cartesian
space leads to a dyadic operator representation associated
with position and momentum variables, for which the
resulting operator-kernel Fourier duality naturally
connects to the ordinary Heisenberg position--momentum
uncertainty principle.
In this sense, the present framework extends the usual notion of
Fourier complementarity from wavefunction amplitudes to mutually
unbiased operator representations in Hilbert--Schmidt space.

With this motivation in place, we now develop the
general operator-frame framework from which the
resulting Fourier dualities and entropic uncertainty
relations emerge.

\section {Continuous Operator Frames and Hilbert--Schmidt Representations}
\noindent Let \(\mathcal H=L^2(\mathbb R)\) be the space of square-integrable functions on $\mathbb{R}$. Throughout the paper, given the continuous-variable setting, the label space may equivalently be viewed as $\mathbb C$ or $\mathbb R^2$, since
$\mathbb C \cong \mathbb R^2$ as real vector spaces. While the complex
notation is natural for operator formulations in quantum
mechanics, we shall frequently use the equivalent real two-dimensional
representation when discussing bilinear kernels and symplectic
structures. 

We consider continuous families of operators
$\{O(\lambda)\}$ and
$\{P(\beta)\}$ with $\lambda,\beta\in\mathbb R^2$ forming generalized orthogonal
continuous operator frames, satisfying
\begin{equation}
\mathrm{Tr}\!\left[O^\dagger(\lambda) O(\lambda')\right]
=C_O\delta^{(2)}(\lambda - \lambda'),~~~\mathrm{Tr}\!\left[P^\dagger(\beta) P(\beta')\right]
=C_P\delta^{(2)}(\beta - \beta').
\end{equation}
for some constant $C_O, C_P > 0$. 
We shall study representations of Hilbert--Schmidt operators
$A\in B_2(\mathcal H)$, where $B_2(\mathcal H)
$
is the space of Hilbert–Schmidt operators on the Hilbert space $\mathcal H$. This space is equipped with the inner product
\begin{equation}
\langle A, B \rangle := \mathrm{Tr}(A^\dagger B),
\end{equation}
which induces the norm $\|A\|_{\mathrm{HS}}^2 = \mathrm{Tr}(A^\dagger A)$.
Any such operator $A$ then admits the expansion
\begin{equation}\label{expA}
A = \frac{1}{C_O} \int_{\mathbb{R}^2} a_O(\lambda)\, O(\lambda)\, d^2\lambda=\frac{1}{C_P} \int_{\mathbb{R}^2} a_P(\beta)\, P(\beta)\, d^2\beta
\end{equation}
where the coefficient function $a_O(\lambda),a_P(\beta) \in L^2(\mathbb{R}^2)$ are given by
\begin{equation}\label{coeff}
a_O(\lambda) = \mathrm{Tr}\!\left[O^\dagger(\lambda) A\right],~~a_P(\beta) = \mathrm{Tr}\!\left[P^\dagger(\beta) A\right].
\end{equation}
The coefficient functions \(a_O(\lambda)\) and \(a_P(\beta)\) constitute two alternative representations of the same Hilbert--Schmidt operator. To quantify the extent to which these representations may be simultaneously localized, we seek an entropic uncertainty relation between the corresponding coefficient distributions. We begin by introducing the normalized coefficient amplitudes
\(f_O(\lambda)\) and \(f_P(\beta)\) as
\begin{equation}\label{fofp}
f_O(\lambda)
=
\frac{a_O(\lambda)}
{\sqrt{C_O}\|A\|_{\mathrm{HS}}},
\qquad
f_P(\beta)
=
\frac{a_P(\beta)}
{\sqrt{C_P}\|A\|_{\mathrm{HS}}},
\end{equation}
so that the corresponding coefficient densities, $|f_O(\lambda)|^2$ and $|f_P(\beta)|^2$ are normalized,
\begin{equation}\label{endpoint2}
\int_{\mathbb R^2}|f_O(\lambda)|^2\,d^2\lambda
=
\int_{\mathbb R^2}|f_P(\beta)|^2\,d^2\beta
=
1.
\end{equation}
The two normalized coefficient functions are related by the transformation induced by the change of operator frame. Let us define an operator $T$, such that $f_P=Tf_O$. Making use of equations (\ref{expA}) and (\ref{coeff}) we write,
\begin{equation}
f_P(\beta)
=
(Tf_O)(\beta)
=
\int_{\mathbb R^2}
K(\beta,\lambda)\,
f_O(\lambda)\,
d^2\lambda,
\label{transform}
\end{equation}
with normalized kernel
\begin{equation}\label{kernel}
K(\beta,\lambda)
=
\frac{
\operatorname{Tr}
\!\left[
P^\dagger(\beta)O(\lambda)
\right]
}
{\sqrt{C_OC_P}}.
\end{equation}
For a well-defined and uniformly bounded $K(\beta,\lambda)$, i.e.,
$
\sup_{\beta,\lambda\in\mathbb R^2}
|K(\beta,\lambda)|
<\infty,
$
we commit to an interpolation framework not unlike that used in deriving
Hirschman's entropy inequality for the Fourier transform, but carried
out in the more general setting of operator-frame transforms. We establish
$L^1\rightarrow L^\infty$ and $L^2\rightarrow L^2$ bounds for the
transform $T$, interpolate these bounds via the Riesz--Thorin theorem \cite{ReedSimon2} to
obtain a Hausdorff--Young-type inequality, and subsequently
differentiate the resulting logarithmic norm inequality at
$p=q=2$ to derive a Shannon entropy relation. The details of this is provided in the Appendix.

The resulting entropic inequality is given by 
\begin{equation}
H_O
+
H_P
\ge
-2\ln \left[\sup_{\beta,\lambda}
|K(\beta,\lambda)|\right ],
\end{equation}
where
$H_O$ and $H_P$ denote the Shannon differential entropy (when finite), given by 
\begin{equation}\label{HP}
H_O
=
-\int_{\mathbb R^2}
|f_O(\lambda)|^2\ln |f_O(\lambda)|^2\,
d^2\lambda,~~
H_P
=
-\int_{\mathbb R^2}
|f_P(\beta)|^2\ln |f_P(\beta)|^2\,
d^2\beta.
\end{equation}
The inequality therefore expresses an operator-level
entropic uncertainty principle: localization of an operator
in one operator basis necessarily implies delocalization in
the other. In particular, no
Hilbert--Schmidt operator can possess arbitrarily sharp
coefficient distributions simultaneously in both operator representations. In this sense,
the uncertainty principle emerges not at the level of
observables or quantum states alone, but intrinsically at
the level of operator representations themselves and is determined by the overlap of the operators.

The operator frames also admit a natural tomographic interpretation. Quite generally, each continuous operator frame may be operationally viewed as defining a distinct measurement or probing scheme for operators, corresponding to a distinct experimental setup, with the expansion coefficients, e.g. $a_O(\lambda)$, providing the corresponding tomographic components for reconstruction. The normalized coefficient distributions, e.g. $|f_O(\lambda)|^2$, characterize how the Hilbert--Schmidt content of $A$ is distributed across the representation. Thus, two frames provide incompatible operational perspectives on the same physical object, and the entropic uncertainty relation quantifies the corresponding minimal uncertainty, or noise floor, for jointly accessing both descriptions. The specific operational realization depends on the physical nature of $A$ and on the choice of operator frames. 

For a quantum state, where $A=\rho$ is a density operator, the coefficient functions associated with a given operator frame provide an alternative tomographic representation of the state. In particular, the choice of frame determines how the state may be probed and reconstructed, while the corresponding normalized coefficient distributions quantify how the state is distributed over the chosen operator representation.
More generally, the same framework may be applied to quantum processes through an appropriate operator representation, such as a Choi representation, whenever the resulting operator belongs to the relevant Hilbert--Schmidt space and the continuous family of operators used for its representation forms a suitable operator frame. In this case, the corresponding coefficient functions provide a process-tomographic representation, with the entropic uncertainty relation quantifying the trade-off between complementary process representations.

\section{Mutually Unbiased Operator Frames}

\noindent We now show that when the operator families satisfy a mutual unbiasedness condition, the associated transform acquires the structure of a generalized Fourier transform. This additional structure permits the use of Fourier-analytic techniques and leads to a distinguished class of entropic uncertainty relations.

Let us now assume that the operator families $\{O(\lambda)\}$ and
$\{P(\beta)\}$ are mutually unbiased in the sense that
their trace overlaps, whenever well-defined, satisfy,
\begin{equation}\label{mub}
\left|\mathrm{Tr}[P^\dagger(\beta)\, O(\lambda)]\right| = C, \quad \forall \lambda,\beta \in \mathbb{R}^2.
\end{equation}
This implies that the overlap can be written as
\begin{equation}
\mathrm{Tr}[P^\dagger(\beta)\, O(\lambda)]
=
C\, e^{-i\phi(\beta,\lambda)},
\end{equation}
for some real-valued phase function $\phi(\beta,\lambda)$. The trace overlap appearing in equation (\ref{mub}) admits a natural
interpretation in terms of the action of the operator
families on the underlying Hilbert space
\(\mathcal H=L^2(\mathbb R)\). Recall that for an orthonormal basis $\{e_n\}\subset\mathcal H$, we can write
\begin{equation}
\mathrm{Tr}[P^\dagger(\beta)O(\lambda)]
=
\sum_n
\langle P(\beta)e_n,\,O(\lambda)e_n\rangle.
\end{equation}
Likewise, in cases involving a continuous resolution of the identity on
$\mathcal H$,
\[
I=\int |\xi\rangle\langle\xi|\,d\xi,
\]
one formally has
\begin{equation}
\mathrm{Tr}[P^\dagger(\beta)O(\lambda)]
=
\int
\langle \xi|
P^\dagger(\beta)O(\lambda)
|\xi\rangle
\,d\xi.
\end{equation}
The trace pairing may therefore be understood as the total overlap
between the transformed functions generated by the two operator
families, summed or integrated over a complete resolution of the
underlying Hilbert space. In this sense, the mutual unbiasedness
condition expresses the fact that the absolute value of this overlap
possesses a constant magnitude independent of $\lambda$ and $\beta$.
This reflects the idea that no element of one operator family is
preferentially aligned with any element of the other.

As the parameters $\lambda$ and $\beta$ serve as labels, or coordinates,
on the homogeneous space $\mathbb{R}^2$, no point in this space is
distinguished a priori. It is therefore natural to require that the
relation between the operator families respects this homogeneity.
Accordingly, we impose covariance under translations of the labels,
namely that for any $\lambda\in\mathbb{R}^2$, a shift $\lambda \mapsto \lambda+\lambda_0$ modifies the
overlap only by a phase factor,
\begin{equation}
\mathrm{Tr}[P^\dagger(\beta) O(\lambda+\lambda_0)]
=e^{-i\chi_\beta(\lambda_0)}\,
\mathrm{Tr}\!\left[P^\dagger(\beta)\,O(\lambda)\right],
\end{equation}
with an analogous condition for translations in $\beta$.
We should emphasise that this choice is not claimed to be unique; rather, it is
the simplest symmetry-compatible assumption ensuring that the relation
between the operator families depends only on relative structure in the
label space. Assuming sufficient regularity of the phase function
\(\phi(\beta,\lambda)\), we have the following proposition.

\begin{prop}
Let $\{O(\lambda)\}$ and
$\{P(\beta)\}$ with $\lambda,\beta\in\mathbb R^2$ be two continuous operator families satisfying
\[
\mathrm{Tr}\!\left[P^\dagger(\beta)O(\lambda)\right]
=
Ce^{-i\phi(\beta,\lambda)},
~ C>0.
\]
Suppose that the phase function
$\phi(\beta,\lambda)$ is sufficiently regular (for instance
continuous in each label) and furthermore that their overlap is covariant under translations of
the labels, in the sense that
\begin{align}
\mathrm{Tr}\!\left[P^\dagger(\beta)O(\lambda+\lambda_0)\right]
&=
e^{-i\chi_\beta(\lambda_0)}
\mathrm{Tr}\!\left[P^\dagger(\beta)O(\lambda)\right],\\
\mathrm{Tr}\!\left[P^\dagger(\beta+\beta_0)O(\lambda)\right]
&=
e^{-i\eta_\lambda(\beta_0)}
\mathrm{Tr}\!\left[P^\dagger(\beta)O(\lambda)\right],
\end{align}
where \(\chi_\beta(\lambda_0)\) is independent of
\(\lambda\) and \(\eta_\lambda(\beta_0)\) is independent
of \(\beta\). Then the overlap necessarily assumes the form
\[
\mathrm{Tr}[P^\dagger(\beta)O(\lambda)]
=
Ce^{-iB(\beta,\lambda)},
\]
where \(B:\mathbb R^2\times\mathbb R^2\to\mathbb R\)
is bilinear.
\end{prop}
\begin{proof}
Given $\mathrm{Tr}\!\left[P^\dagger(\beta)O(\lambda)\right]
=
C e^{-i\phi(\beta,\lambda)}$ for some real-valued phase function $\phi(\beta,\lambda)$, the first covariance condition gives
\begin{equation}
C e^{-i\phi(\beta,\lambda+\lambda_0)}
=
e^{-i\chi_\beta(\lambda_0)}
C e^{-i\phi(\beta,\lambda)}.
\end{equation}
Hence, modulo integer multiples of $2\pi$,
\begin{equation}
\phi(\beta,\lambda+\lambda_0)-\phi(\beta,\lambda)
=
\chi_\beta(\lambda_0).
\end{equation}
The right-hand side is independent of $\lambda$, so for fixed $\beta$ the
increment of $\phi$ under translations of $\lambda$ depends only on the
translation. Therefore, under mild regularity assumptions such as continuity, $\phi$ is affine in $\lambda$.
Applying the second covariance condition similarly implies that $\phi$ is
affine in $\beta$. Up to phase factors depending only on one label, which may
be absorbed into redefinitions of $O(\lambda)$ and $P(\beta)$, the remaining
phase is bilinear:
\begin{equation}
\phi(\beta,\lambda)=B(\beta,\lambda).
\end{equation}
Thus
\begin{equation}
\mathrm{Tr}\!\left[P^\dagger(\beta)O(\lambda)\right]
=
C e^{-iB(\beta,\lambda)}.
\end{equation}
\end{proof}
We shall assume throughout that the bilinear pairing \(B\), which emerges from compatibility between mutual unbiasedness and translational covariance on the label space also is non-degenerate. This assumption is motivated by the
observation that, if \(B\) were degenerate, there would
exist nonzero label directions along which the phase factor
\(e^{-iB(\beta,\lambda)}\) becomes independent of the
complementary label variable. The overlap kernel would then
become insensitive to variations along such directions,
leading to a partially trivial overlap structure despite
the constant-modulus condition. The non-degeneracy
condition therefore ensures that the pairing \(B\)
identifies the \(\beta\)- and \(\lambda\)-label spaces as
dual variables: no nonzero \(\beta\) pairs trivially with
all \(\lambda\), and no nonzero \(\lambda\) pairs
trivially with all \(\beta\). Accordingly, we write
\begin{eqnarray}\label{B}
B(\beta,\lambda)=\beta^{T}M\lambda,
\qquad
\det M\neq 0,
\end{eqnarray}
where \(M\) is an invertible real \(2\times2\) matrix. We shall make use of this equation explicitly  later.
Substituting the mutual unbiasedness relation (21) into the kernel definition of equation (\ref{kernel}), we obtain
\begin{equation}
K(\beta,\lambda)
=
\frac{C}{\sqrt{C_O C_P}}
e^{-iB(\beta,\lambda)}.
\end{equation}
We shall continue to denote by $f_O$ and $f_P$ the normalized coefficient amplitudes defined in equation (\ref{fofp}) and the transform (\ref{transform}) relating these amplitudes now assumes the form, 
\begin{equation}
f_P(\beta)
=
\frac{C}{\sqrt{C_O C_P}}
\int_{\mathbb R^2}
e^{-iB(\beta,\lambda)}
f_O(\lambda)\, d^2\lambda .
\label{eq:BFourier}
\end{equation}
Since both coefficient representations arise from Parseval-type expansions of the same Hilbert--Schmidt operator, the corresponding transform must be unitary on $L^2(\mathbb R^2)$, the forward and inverse coefficient maps must compose to the identity. Making use of the nondegeneracy of the bilinear form, we find
\begin{align}
f_O(\lambda)
&=
\left(
\frac{C}{\sqrt{C_O C_P}}
\right)^2
\int_{\mathbb R^2}
\int_{\mathbb R^2}
e^{\,iB(\beta,\lambda)}
e^{-iB(\beta,\lambda')}
f_O(\lambda')
\, d^2\beta\, d^2\lambda'
\nonumber\\
&=
\left(
\frac{C}{\sqrt{C_O C_P}}
\right)^2
\int_{\mathbb R^2}
f_O(\lambda')
\left[
\int_{\mathbb R^2}
e^{\,iB(\beta,\lambda-\lambda')}
\, d^2\beta
\right]
d^2\lambda'.
\end{align}
Writing $B(\beta,\lambda)=\beta^{T}M\lambda$ with
$\det M\neq 0$, one has the Fourier completeness relation
\begin{equation}
\int_{\mathbb R^2}
e^{\,iB(\beta,\lambda-\lambda')}
\, d^2\beta
=
\frac{(2\pi)^2}{|\det M|}
\,\delta^{(2)}(\lambda-\lambda').
\end{equation}
Therefore,
\begin{equation}
f_O(\lambda)
=
\frac{(2\pi)^2}{|\det M|}
\left(
\frac{C}{\sqrt{C_O C_P}}
\right)^2
f_O(\lambda),
\end{equation}
and consistency requires
\begin{equation}
\frac{C}{\sqrt{C_O C_P}}
=
\frac{\sqrt{|\det M|}}{2\pi}.
\label{eq:normalisation_condition}
\end{equation}
The transform between mutually unbiased coefficient amplitudes therefore becomes
\begin{equation}
f_P(\beta)
=
\frac{\sqrt{|\det M|}}{2\pi}
\int_{\mathbb R^2}
e^{-iB(\beta,\lambda)}
f_O(\lambda)\, d^2\lambda,
\label{fP}
\end{equation}
with inverse
\begin{equation}
f_O(\lambda)
=
\frac{\sqrt{|\det M|}}{2\pi}
\int_{\mathbb R^2}
e^{\,iB(\beta,\lambda)}
f_P(\beta)\, d^2\beta.
\end{equation}
An equivalent condition follows by evaluating
the relevant composition acting on \(a_P(\beta)\).
Therefore, once the overlap between the two operator bases is governed by
the bilinear and nondegenerate phase \(e^{-iB(\beta,\lambda)}\), the two coefficient
representations of the same operator are naturally related as Fourier-dual
functions with respect to the bilinear pairing \(B\). In this sense, operator mutual unbiasedness induces a Fourier relation at the level of Hilbert--Schmidt coefficient functions.

In deriving an entropic relation between these two quantities, we first return to equation (\ref{B})
and introduce the variable, $\eta\in \mathbb{R}^2$, given as,
\begin{equation}
\eta=M^T\beta \Rightarrow d^2\eta=|\det M|\,d^2\beta.
\end{equation}
We now introduce the function, $\widehat{f}_O(\eta)$, related to $f_O(\lambda)$ via the standard Fourier
\begin{equation}\label{fhat}
\widehat{f}_O(\eta)
:=
\frac{1}{2\pi}
\int_{\mathbb R^2}
f_O(\lambda)
e^{-i\eta^T\lambda}
\,d^2\lambda,
\end{equation}
with
\begin{equation}
\int_{\mathbb R^2}
|\widehat{f}_O(\eta)|^2
\,d^2\eta=1.
\end{equation}
 We then define
 \begin{equation}
H_{\widehat{O}}
=
-\int_{\mathbb R^2}
|\widehat{f}_O(\eta)|^2\ln |\hat{f}_O(\eta)|^2\,
d^2\eta,
\end{equation}
and
using the Hirschman--Beckner relation \cite{Hirschman1957,Beckner1975,Bialynicki1975} for the Fourier normalization convention of equation (\ref{fhat}) one immediately arrive at,
 \begin{equation}\label{eur}
H_O+H_{\widehat O}
\ge
2\ln(\pi e).
\end{equation}
Thanks to Eq. (\ref{fhat})we can rewrite the Fourier relation relation of equation (\ref{fP}) as
 \begin{equation}
f_P(\beta)
=
\sqrt{|\det M|}\,\widehat f_O(\eta),
\end{equation}
and substituting this into $H_P$ as defined in equation (\ref{HP}), we obtain
\begin{align}
H_P
&=
-\int_{\mathbb R^2}
|\widehat f_O(\eta)|^2
\ln |\widehat f_O(\eta)|^2
\,d^2\eta
-
\ln|\det M|\nonumber\\
&=H_{\widehat O}
-
\ln|\det M|.
\end{align}
Together with equation (\ref{eur}), we obtain our desired entropic relation,
\begin{equation}\label{EUR}
H_O+H_P
\ge
2\ln(\pi e)-\ln|\det M|.
\end{equation}
The inequality (\ref{EUR}) constitutes a refinement of the general
operator-frame entropic uncertainty relation obtained in Sec.~II.
Indeed, for mutually unbiased operator frames,
$
\sup_{\beta,\lambda}|K(\beta,\lambda)|
=
\sqrt{|\det M|}/2\pi,
$
so that the interpolation argument of Sec.~II yields a lower bound
which is exceeded by (\ref{EUR}) by the constant amount
$
2\ln\!\left(e/2\right).
$
This improvement arises from the additional Fourier structure induced
by mutual unbiasedness, which permits the use of the sharp
Hirschman--Beckner theorem rather than interpolation arguments alone.

It is worth noting that the entropic relation of (\ref{EUR}) bears a formal resemblance to continuous-variable entropic uncertainty relations derived for families of observables ~\cite{Cerf1,Cerf2}, where the entropies quantify the uncertainty associated with measurement outcomes and the lower bound is expressed in terms of the determinant of a commutator matrix. In the present framework, however,  the entropies characterize the localization of operator representations in mutually unbiased operator frames rather than the statistics of measurement outcomes. Further to that, the determinant reflects the nondegeneracy of the bilinear overlap kernel defining mutual unbiasedness between operator families, which need not correspond to Hermitian observables. Interestingly, the dimensional contribution $2\ln(\pi e)$ appearing in the phase-space realization coincides formally with the corresponding term arising in the two-dimensional observable setting of continuous-variable entropic uncertainty relations ~\cite{Cerf1,Cerf2}. In fact, the two relations concern fundamentally different quantities: in the present work, the entropy is defined over operator-frame coefficients on a two-dimensional label space, whereas in the observable setting it is defined over measurement probability distributions.

The inequality becomes particularly transparent in the
canonical symplectic case \(|\det M|=1\), for which
the lower bound reduces to \(2\ln(\pi e)\), identical
to the standard two-dimensional Hirschman--Beckner
bound. 

\section{Displacement Operators and Wigner Kernels as a Canonical Realization}

\noindent \noindent We now specialize the general operator-frame construction to the
continuous-variable phase space $(\mathbb R^2,\Omega)$ equipped with its
canonical symplectic structure, where $\Omega$ may be represented, in canonical coordinates, by an antisymmetric matrix. In this setting, the displacement operators
provide a natural continuous orthogonal operator  frame associated with
phase-space translations. As we shall show, the corresponding displaced parity
operators form a mutually unbiased operator family whose overlap kernel is
generated by the symplectic bilinear form. This realizes the general
Fourier--bilinear framework developed earlier in a canonical phase-space
setting and leads naturally to a symplectic Fourier duality between the two
operator representations.

Let $D(\alpha)$ be the displacement operators \cite{CahillGlauber19691},
\begin{eqnarray}
D(\alpha)= e^{\alpha a^\dagger - \alpha^* a}, ~~\alpha \in \mathbb{C} \simeq \mathbb{R}^2.
\end{eqnarray}
and the family of operators $\{D(\alpha)\}$ be a  continuous frame for $B_2(\mathcal{H})$ fulfilling the orthogonality relation
\begin{eqnarray}
\mathrm{Tr}[D(\alpha) D^\dagger(\alpha')]
= \pi \delta^{(2)}(\alpha' - \alpha).
\end{eqnarray}
Let us further consider the case for an orthogonal basis of a family of displaced parity operators, namely the Wigner kernel, $\Delta(\beta), \beta\in \mathbb{C} \simeq \mathbb{R}^2$. The Wigner phase-space representation admits a natural operator-valued kernel formulation; the Wigner function, $W(\beta)$, of a quantum state $\rho$ is expressed as the expectation value of a displaced parity operator,
 \(
W(\beta)=\mathrm{Tr}[\rho\,\Delta(\beta)]
\),
\cite{CahillGlauber19692,Royer1977}.
The  Wigner kernel is Hermitian and unitary, with eigenvalues $\pm1$, and corresponds physically to a local parity measurement about $\beta$ and is defined as
\begin{equation}
\Delta(\beta) := 2D(\beta)\,\Pi\,D^\dagger(\beta),
\end{equation}
where $\Pi$ is the parity operator. 
The family $\{\Delta(\beta)\}$ forms an orthogonal, continuous operator  frame and any operator can therefore be written in terms of the $\{\Delta(\beta)\}$ \cite{CahillGlauber19692}. 

Let us consider the trace overlap, 
\begin{eqnarray}\label{overlap}
\mathrm{Tr}\!\left[\Delta^\dagger(\beta)D(\alpha)\right]=2e^\kappa\mathrm{Tr}\!\left[D(\alpha)\Pi\right]
\end{eqnarray} 
where $\kappa=\alpha\beta^*-\alpha^*\beta=2i\,\operatorname{Im}(\alpha^*\beta)$. Then,
\begin{align}
\operatorname{Tr}[D(\alpha)\Pi]
&=\int_{\mathbb R^2}\frac{d^2\zeta}{\pi}\,
\langle \zeta|D(\alpha)\Pi|\zeta\rangle
=\int_{\mathbb R^2}\frac{d^2\zeta}{\pi}\,
\langle \zeta|D(\alpha)|-\zeta\rangle. \\
&=\int_{\mathbb R^2}\frac{d^2\zeta}{\pi}\,\exp\!\left(-\frac{|\alpha|^2}{2}\right)\,
\exp\!\left(-2|\zeta|^2+\alpha\zeta^*+\alpha^*\zeta\right)=\frac{1}{2}.
\end{align}
Hence from equation (\ref{overlap}),
\begin{eqnarray}
|\mathrm{Tr}\!\left[\Delta^\dagger(\beta)D(\alpha)\right]|=|2e^\kappa\mathrm{Tr}\!\left[D(\alpha)\Pi\right]|=|e^\kappa|=1.
\end{eqnarray} 
Thus, the families $\{D(\alpha)\}$ and $\{\Delta(\beta)\}$ form a pair of
mutually unbiased continuous operator frames as per our definition of equation (\ref{mub}) with the phase kernel being not only bilinear but symplectic.
In the present case, the bilinear phase appearing in the overlap kernel is
the symplectic form on $\mathbb R^2$. With
$\boldsymbol{\alpha}=(\operatorname{Re}\alpha,\operatorname{Im}\alpha)^{T},
\boldsymbol{\beta}=(\operatorname{Re}\beta,\operatorname{Im}\beta)^{T}$, one has
\begin{equation}
\kappa
=
\alpha\beta^*-\alpha^*\beta
=-2i\boldsymbol{\alpha}^{T}
\Omega
\boldsymbol{\beta},
\end{equation}
where $\det \Omega = 1$. Equivalently, the bilinear kernel in this case, $B$ may be written in the form
$B(\alpha,\beta)=\boldsymbol{\alpha}^T M\boldsymbol{\beta}$ with $M=-2\Omega$, so that
$|\det M|=4$. The overlap kernel is therefore generated by a
nondegenerate antisymmetric bilinear form associated with the canonical
symplectic structure underlying continuous-variable phase space.

\subsection{Entropic uncertainty between displacement and Wigner kernel}

\noindent Let $A_{D\Delta}\in B_2(\mathcal{H})$ be a Hilbert-Schmidt operator, such that 
\begin{align}
A_{D\Delta} = \frac{1}{\pi} \int_{\mathbb R^2} a_D(\alpha)\, D(\alpha)\, d^2\alpha 
= \frac{1}{\pi} \int_{\mathbb R^2} a_\Delta(\beta)\, \Delta(\beta)\, d^2\beta.
\end{align}
with the coefficient functions, $a_D(\alpha)$ and $a_\Delta(\beta)$
\begin{align}\label{coef}
a_D(\alpha) := \mathrm{Tr}\!\left[D^\dagger(\alpha)A_{D\Delta}\right],~~
a_\Delta(\beta) := \mathrm{Tr}\!\left[\Delta^\dagger(\beta)A_{D\Delta}\right].
\end{align}
The coefficient intensity $|a_D(\alpha)|^2$ defines a measure on $\mathbb R^2$ whose total mass is $\pi \|A\|_{\mathrm{HS}}^2$.
Let us define the normalized coefficient densities
\begin{align}
f_D(\alpha) := \frac{a_D(\alpha)}{\sqrt{\pi}\|A_{D\Delta}\|_{\mathrm{HS}}},~~f_\Delta(\beta) := \frac{a_\Delta(\beta)}{\sqrt{\pi}\|A_{D\Delta}\|_{\mathrm{HS}}},
\end{align}
which satisfy
\begin{equation}
\int_{\mathbb R^2} |f_D(\alpha)|^2\,d^2\alpha
=
\int_{\mathbb R^2} |f_\Delta(\beta)|^2\,d^2\beta
=
1.
\end{equation}
The corresponding Shannon differential entropies
\[
H_D:=-\int_{\mathbb C}|f_D(\alpha)|^2\log |f_D(\alpha)|^2\,d^2\lambda,
\qquad
H_\Delta:=-\int_{\mathbb C}|f_\Delta(\beta)|^2\log |f_\Delta(\beta)|^2\,d^2\beta,
\]
together with $\det M=4$, applying equation (\ref{EUR}), give,
\[
H_D+H_\Delta \ge 2\ln\!\left(\frac{\pi e}{2}\right).
\]
The above uncertainty relation refers to the spreads of the coefficient distributions of $A_{D\Delta}$ in the corresponding representations. The inequality therefore implies that an operator cannot be simultaneously sharply concentrated in both the displacement and Wigner-kernel representations, reflecting the Fourier duality between these operator frames.

\section{Position and momentum operators as mutually unbiased frames}

\noindent The general operator-frame construction may also be realized
through the operator families generated by the
position and momentum eigenstates.

Let $\{|x\rangle\}_{x\in\mathbb R}$ and
$\{|p\rangle\}_{p\in\mathbb R}$ denote the generalized
eigenbases of the canonical position and momentum operators, understood in the usual rigged-Hilbert-space sense,
satisfying
\begin{equation}
\langle x|x'\rangle = \delta(x-x'),
\qquad
\langle p|p'\rangle = \delta(p-p'),
\end{equation}
We now consider the operator families, $\{X(x,x')\}$ and $\{Y(p,p')\}$,
\begin{equation}
X(x,x') := |x\rangle\langle x'|,
\qquad
Y(p,p') := |p\rangle\langle p'|,
\end{equation}
which form generalized continuous operator  frames for the Hilbert--Schmidt space $B_2(\mathcal{H})$.
Their orthogonality relations are
\begin{eqnarray}
\mathrm{Tr}
\!\left[
X^\dagger(x,x')
X(y,y')
\right]
&=&
\delta(x-y)\delta(x'-y'),
\\
\mathrm{Tr}
\!\left[
Y^\dagger(p,p')
Y(q,q')
\right]
&=&
\delta(p-q)\delta(p'-q').
\end{eqnarray}
Let us now evaluate the overlap between the two operator
families. 
\begin{eqnarray}
\mathrm{Tr}
\!\left[
Y^\dagger(p,p')
X(x,x')
\right]
&=&
\mathrm{Tr}
\!\left[
|p'\rangle\langle p|
x\rangle\langle x'|
\right]
\\
&=&
\langle p|x\rangle
\langle x'|p'\rangle.
\end{eqnarray}

Using the Fourier kernel relation between position and
momentum eigenstates,
\begin{equation}
\langle p|x\rangle
=
\frac{1}{\sqrt{2\pi}}
e^{-ipx},
\qquad
\langle x'|p'\rangle
=
\frac{1}{\sqrt{2\pi}}
e^{ip'x'},
\end{equation}
we obtain
\begin{equation}
\mathrm{Tr}
\!\left[
Y^\dagger(p,p')
X(x,x')
\right]
=
\frac{1}{2\pi}
e^{-i(px-p'x')}.
\end{equation}

Hence,
\begin{equation}
\left|
\mathrm{Tr}
\!\left[
Y^\dagger(p,p')
X(x,x')
\right]
\right|
=
\frac{1}{2\pi},
\end{equation}
independent of $(x,x')$ and $(p,p')$.
The two operator families therefore form a pair of mutually
unbiased continuous operator frames in the sense of
equation (5).
The overlap kernel is generated by the bilinear form, $B$, given as
\begin{equation}
B\!\left[
(p,p'),(x,x')
\right]
=
px-p'x',
\end{equation}
which may be written in matrix form as
\begin{equation}
B(\beta,\lambda)
=
(p,p') M (x,x')^T,~~
\text{with} ~M =
\begin{pmatrix}
1 & 0 \\
0 & -1
\end{pmatrix}
\end{equation}
being a symmetric matrix of determinant $-1$.

\subsection{Position-Momentum Uncertainty Principle}
\noindent For a Hilbert--Schmidt operator $A_{XY}\in B_2(\mathcal H)$,
we may write
\begin{equation}
A_{XY} =
\int_{\mathbb R^2}
a_X(x,x')\, |x\rangle\langle x'|\, dx\,dx'
=
\int_{\mathbb R^2}
a_Y(p,p')\, |p\rangle\langle p'|\, dp\,dp',
\end{equation}
where
\begin{equation}
a_X(x,x') = \mathrm{Tr}\!\left[X^\dagger(x,x')A_{XY}\right],
\qquad
a_Y(p,p') =
\mathrm{Tr}\!\left[P^\dagger(p,p')A_{XY}\right].
\end{equation}
Introducing the normalized coefficient amplitudes
\begin{equation}
 f_X(x,x')=
\frac{a_X(x,x')}{\|A_{XY}\|_{\mathrm{HS}}},
\qquad
f_Y(p,p')=
\frac{a_Y(p,p')}{\|A_{XY}\|_{\mathrm{HS}}},
\end{equation}
one has
\begin{equation}
\int_{\mathbb R^2}|f_X(x,x')|^2\,dx\,dx'
=
\int_{\mathbb R^2}|f_Y(p,p')|^2\,dp\,dp'
=1.
\end{equation}
The associated differential entropies are
\begin{align}
H_X
&=
-\int_{\mathbb R^2}
|f_X(x,x')|^2
\ln |f_X(x,x')|^2\,dx\,dx',
\nonumber\\
H_Y
&=
-\int_{\mathbb R^2}
|f_Y(p,p')|^2
\ln |f_Y(p,p')|^2\,dp\,dp'.
\end{align}
Making use of the general entropic relation of equation (\ref{EUR}) with $|\det M|=1$, we have
\begin{equation}
H_X+H_Y\geq 2\ln(\pi e).
\end{equation}
This inequality states that a Hilbert--Schmidt operator cannot
be simultaneously sharply localized in both the position-dyadic representation and the momentum-dyadic
representation. 

\subsection{Pure-State Density Operators and the Connection to Position--Momentum Uncertainty}

\noindent It is immediately tempting to suspect a relationship between the
entropic uncertainty relation derived above and the ordinary
position--momentum uncertainty principle. To clarify the connection between the present operator-kernel
formalism and the ordinary position--momentum uncertainty principle,
consider a pure quantum state
$\rho = |\psi\rangle\langle\psi|.
$
In the position representation,
$
\psi(x)=\langle x|\psi\rangle,
$
so that
\begin{equation}
|\psi\rangle
=
\int_{\mathbb R}\psi(x)|x\rangle\,dx,
\qquad
\langle\psi|
=
\int_{\mathbb R}\psi^*(x')\langle x'|\,dx'.
\end{equation}
The associated density operator therefore admits the expansion
\begin{equation}
\rho
=
\iint_{\mathbb R^2}
\psi(x)\psi^*(x')
\,|x\rangle\langle x'|
\,dx\,dx'.
\end{equation}
Hence, the coefficient kernel in the position-dyadic representation is
given by
\begin{equation}
\rho(x,x')
=
\langle x|\rho|x'\rangle
=
\psi(x)\psi^*(x').
\end{equation}
Likewise, in the momentum representation,
\begin{equation}
\widetilde{\rho}(p,p')
=
\langle p|\rho|p'\rangle
=
\widetilde{\psi}(p)\widetilde{\psi}^{*}(p'),
\end{equation}
where
\begin{equation}
\widetilde{\psi}(p)
=
\frac{1}{\sqrt{2\pi}}
\int_{\mathbb R}
\psi(x)e^{-ipx}\,dx.
\end{equation}
Substituting the Fourier transform relation for
$\widetilde{\psi}(p)$ and its complex conjugate yields
\begin{equation}
\widetilde{\rho}(p,p')
=
\frac{1}{2\pi}
\iint_{\mathbb R^2}
\rho(x,x')
\,e^{-i(px-p'x')}
\,dx\,dx'.
\end{equation}
Thus, the dyadic Fourier relation appearing in the
present framework is inherited directly from the
ordinary Fourier relation between the position- and
momentum-space wavefunction amplitudes.
The diagonal components of the kernels recover the ordinary position
and momentum probability densities,
\begin{equation}
\rho(x,x)
=
|\psi(x)|^2,
\qquad
\widetilde{\rho}(p,p)
=
|\widetilde{\psi}(p)|^2.
\end{equation}
Accordingly, the standard position--momentum entropic uncertainty
relation ultimately originates from the Fourier duality between the
position-space and momentum-space wavefunction amplitudes
$\psi(x)$ and $\widetilde{\psi}(p)$, with the associated probability
densities $|\psi(x)|^2$ and $|\widetilde{\psi}(p)|^2$ arising through
the Born rule. In contrast, the present operator-level framework concerns
the localization properties of the complete density-operator
representations
$
\rho(x,x')$ and
$\widetilde{\rho}(p,p'),
$
including both diagonal probability contributions and
off-diagonal coherence terms. The resulting entropic uncertainty relation therefore constrains not
only the simultaneous localization of position and momentum
probability distributions, but also the localization of the
corresponding coherence structure encoded in the off-diagonal operator
components.



\section{Conclusion}
\noindent The present work develops a general operator-frame framework for entropic uncertainty relations in Hilbert--Schmidt space. Beginning with continuous operator frames and their associated coefficient representations, we established a general entropic uncertainty relation by combining endpoint norm estimates with Riesz--Thorin interpolation. The entropic relation quantifies the extent to which two operator representations may be simultaneously localized.

We then identified a distinguished class of operator frames satisfying a mutual unbiasedness condition. 
Under natural assumptions of translational covariance and
nondegeneracy, the overlap structure between such frames induces a
bilinear Fourier transform relating their normalized coefficient
amplitudes. The resulting construction unifies a variety of Fourier
dualities encountered in continuous-variable quantum theory within a single framework, with different physical realizations arising from different choices of bilinear pairing.While arbitrary operator frames obey the general entropic uncertainty relation derived from interpolation theory, mutually unbiased operator frames possess a Fourier structure to satisfy a strictly stronger Hirschman--Beckner-type uncertainty relation. In this sense, mutually unbiased operator frames occupy a special position within the broader operator-frame framework, providing a natural operator-level analogue of Fourier duality and complementarity.

The displacement operators and Wigner kernels furnish a canonical continuous-variable realization of the construction. In this setting, the familiar relationship between characteristic-function and Wigner representations emerges as a concrete instance of the general theory, with the symplectic structure encoded through the bilinear pairing underlying the Fourier transform. Likewise, the Cartesian dyadic construction based on position and momentum eigenstates illustrates how the resulting operator-level framework naturally encompasses the ordinary position--momentum Fourier duality while extending beyond probability distributions to the full operator kernels, including their coherence structure.

A natural extension for future investigation is the exploration of additional realizations of mutually unbiased operator frames beyond the two-dimensional label spaces considered in the phase-space and Cartesian constructions presented here. More generally, one may consider operator frames indexed by higher-dimensional parameter spaces (corresponding for example, to multimode continuous-variable systems) or generated from alternative families of continuous-variable eigenstates. Such extensions may reveal new forms of operator-level complementarity and associated entropic uncertainty relations.
An interesting future direction would also include the construction of operator frames from rotated quadrature eigenstates \cite{Cerf1,Cerf2}. Since distinct quadrature eigenbases possess constant-modulus overlaps, the associated dyadic operator families may provide natural candidates for further realizations of the mutual unbiasedness structure considered here. In such settings, the diagonal sectors of the corresponding operator kernels yield the probability distributions associated with quadrature measurements. It would therefore be of interest to investigate whether observable-level entropic uncertainty relations, such as those obtained for rotated quadratures \cite{Cerf1,Cerf2}, can be understood as arising from diagonal restrictions of a more general operator-frame framework. Such a connection would further illuminate the relationship between conventional uncertainty relations for measurement outcomes and the broader operator-level complementarity developed in the present work.

\section{Appendix}
\noindent The generalized orthogonality and reconstruction relations in Eqs.~(1) and (3), together with the trace pairings in Eq.~(4), define the coefficient maps
\begin{equation}
A \longmapsto a_O(\lambda)
= \operatorname{Tr}\!\left[O^\dagger(\lambda)A\right],
\qquad
A \longmapsto a_P(\beta)
= \operatorname{Tr}\!\left[P^\dagger(\beta)A\right],
\end{equation}
which are rescaled isometries from $B_2(\mathcal{H})$ into $L^2(\mathbb{R}^2)$ via the corresponding Parseval relations,
\begin{equation}
\|A\|_{\mathrm{HS}}^2
=
\frac{1}{C_O}
\int_{\mathbb{R}^2}
|a_O(\lambda)|^2\,d^2\lambda
=
\frac{1}{C_P}
\int_{\mathbb{R}^2}
|a_P(\beta)|^2\,d^2\beta.
\end{equation}
This also ensures the normalization of $f_O$ and $f_P$ as defined in Eq.~(5). We restrict attention to coefficient functions \(f_O\) for which the differential Shannon entropy exists and which admit an approximating sequence of test functions,\(\{\phi_O^{(n)}\}\subset\mathcal{D}$, where
$\mathcal{D}$ is a dense subspace contained in $ L^1(\mathbb{R}^2)\cap L^2(\mathbb{R}^2)\) satisfying
\begin{equation}
\|\phi_O^{(n)}-f_O\|_2\to0.
\end{equation}
The \(L^2\) convergence ensures that the approximating sequence represents the same limiting coefficient function. Since $T$ maps $\mathcal{D}$ into itself, we could define a corresponding sequence via, $T\phi_O^{(n)}=\phi_P^{(n)},$ as $\{\phi_P^{(n)}\}_{n=1}^{\infty}\subset\mathcal{D}$. In order to apply the Riesz--Thorin interpolation theorem \cite{ReedSimon2}, we consider the representative $\phi_O^{(n)}\in \mathcal{D}$ (and the corresponding $\phi_P^{(n)}$), and we begin with estimating an endpoint via  eq.~\eqref{transform},
\begin{align}
|\phi_P^{(n)}(\beta)|
&=
\left|
\int
K(\beta,\lambda)
|\phi_O^{(n)}(\lambda)\,
d^2\lambda
\right|
\nonumber\le
\int
|K(\beta,\lambda)|
|\phi_O^{(n)}(\lambda)|
\,d^2\lambda
\nonumber\\
&\le
C_K
\int
|\phi_O^{(n)}(\lambda)|
\,d^2\lambda
=C_K\|\phi_O^{(n)}\|_1,
\end{align}
with $C_K=\sup_{\beta,\lambda}|K(\beta,\lambda)|.$ Taking the supremum over $\beta$ gives
$\|T\phi_O^{(n)}\|_\infty \le C_K\,\|\phi_O\|_1.$
This establishes one endpoint bound. Equation (\ref{endpoint2}) which gives $\|\phi_O^{(n)}\|_2=\|\phi_P^{(n)}\|_2$, establishes another endpoint bound,
$\|T\phi_O^{(n)}\|_2
=
\||\phi_O^{(n)}\|_2$.
%
Applying the Riesz--Thorin interpolation theorem \cite{ReedSimon2} to the endpoint bounds
yields
\begin{equation}
\||\phi_P^{(n)}\|_q
\le
C_K^{\,1-\theta}
\,
\|\phi_O^{(n)}\|_p,
\end{equation}
where
\begin{equation}
\frac{1}{p}
=
1-\frac{\theta}{2},
\qquad
\frac{1}{q}
=
\frac{\theta}{2},
\qquad
0\le\theta\le1.
\end{equation}
Eliminating $\theta$, these relations give
$q=p/(p-1)$ and $1-\theta=2/p-1$. Hence
\begin{equation}
\|\phi_P^{(n)}\|_q
\le
C_K^{\,\frac{2}{p}-1}
\,
\|\phi_O^{(n)}\|_p,
\qquad
1\le p\le2,
\qquad
\frac1p+\frac1q=1.
\label{eq:appendix-HY}
\end{equation}
%
%
Taking logarithms of Eq.~\eqref{eq:appendix-HY} gives
{\begin{equation}
\ln \|\phi_P^{(n)}\|_q
\le
\left(\frac{2}{p}-1\right)\ln C_K
+
\ln \||\phi_O^{(n)}\|_p ,
\label{eq:logHY}
\end{equation}}
with $q=p/(p-1)$.
%
%
%
%
Differentiating Eq.~\eqref{eq:logHY} with respect to $p$
and evaluating at $p=q=2$ yields
\begin{equation}
{H(|\phi_O^{(n)}|^2)
+
H(|\phi_P^{(n)}|^2)
\ge
-2\ln C_K,}
\end{equation} where
$H$ denotes the Shannon differential entropy. {We note that we have assumed the finiteness of the Shannon entropy for the test functions. The regular coefficient functions considered here include a broad class of smooth and well-behaved distributions that arise in physically relevant representations.
To extend the result to a
general coefficient function $f_O \in L^2(\mathbb{R}^2)$, however,
$L^2$-convergence alone is insufficient to guarantee convergence of the
differential Shannon entropy. Such convergence is ensured under
additional regularity or integrability conditions on the coefficient
densities and their approximating sequences; several sufficient
conditions are known in the literature \cite{GodavartiHero2004,PieraParada2009,
GhourchianGohariAmini2017,Godavarti2026}. As a particular example, Piera and Parada \cite{PieraParada2009} provide sufficient conditions based on convergence in variation together with appropriate boundedness conditions on the probability densities, providing one possible route to the required entropy convergence. Accordingly, the extension to
$f_O$ is understood for coefficient functions and approximating
sequences for which
\begin{equation}
H\!\left(|\phi_O^{(n)}|^2\right)
\longrightarrow
H\!\left(|f_O|^2\right),
\end{equation}
and analogously for the $P$-representation.
}

%


\end{document}